\documentclass[11pt]{llncs}

\usepackage{fullpage}
\usepackage{epsfig}
\usepackage{amssymb}

\begin{document}

\mainmatter 

\title{A 3-approximation algorithm for computing a parsimonious first speciation in the gene duplication model} 

\author{Cedric Chauve\inst{1,2} \and Aida Ouangraoua\inst{1,2,3}}

\institute{ Department of Mathematics, Simon Fraser University,
  Burnaby (BC), Canada \and LaCIM, UQAM, Montr\'eal (QC), Canada \and
  LaBRI, Universit\'e Bordeaux I, Talence, France \\
  \email{cedric.chauve@sfu.ca, aida.ouangraoua@sfu.ca}}


\maketitle

\begin{abstract}
  We consider the following problem: from a given set of gene families
  trees on a set of genomes, find a first speciation, that splits these
  genomes into two subsets, that minimizes the number of gene
  duplications that happened before this speciation. We call this
  problem the Minimum Duplication Bipartition Problem.  Using a
  generalization of the Minimum Edge-Cut Problem, known as Submodular
  Function Minimization, we propose a polynomial time and space
  3-approximation algorithm for the Minimum Duplication Bipartition Problem.
\end{abstract}

\section{Introduction}
\label{sec:intro}

Gene duplication is an evolutionary mechanism, that played an major
role in the evolution of the genomes of important groups of eukaryotes
such as vertebrates~\cite{BLOMME-GB7},
insects~\cite{HAHN-PLOSGENETICS3}, plants~\cite{SANDERSON-BMCEB7} or
fungi~\cite{WAPINSKI-NATURE07}. Gene duplications, together with gene
losses, results in {\em gene families}, that can contain several
copies of a same gene in a given genome. Recent progresses in methods
for reconstructing phylogenetic trees for individual gene families,
{\em gene trees}, have resulted in large sets of accurate gene trees
for eukaryote species. Phylogenomics aims at reconstructing the
evolution of {\em species} (genomes) by inferring a species tree, for
a set of genomes, from a set of gene trees. The {\em Minimum
Duplication Problem} asks to find, from a set of gene trees, a species
tree that induces an evolutionary history with a minimum number of
gene duplications. It has been applied on several eukaryotic datasets
with success (see~\cite{SANDERSON-BMCEB7,WEHE-BIOINFO24} for
example). The Minimum Duplication Problem is
NP-hard~\cite{MA-SIAMJC30}, but it can be solved by a fixed-parameter
algorithm, based on a parameter relevant on true
datasets~\cite{HALLETT-RECOMB00}, and recent work on local-search
heuristics have proved to be efficient to process large
datasets~\cite{BANSAL-RECOMB07,WEHE-BIOINFO24}.

Recently in~\cite{CHAUVE-RECOMB09,SCORNA-LATA09}, a formal link
between the Minimum Duplication Problem and the problem of
reconstructing supertrees~\cite{BININDA-EDMONDS-2004} has been
introduced. In the supertree problem, given a set of gene trees from
orthologous genes (at most one member of each family is conserved in
each considered genome), the goal is to reconstruct a species tree
that agrees with the maximum number of gene trees. This problem is
NP-hard too, even in the simple cases where each gene tree contains
only three leaves~\cite{BRYANT-1997} or each gene tree contains a
single internal vertex~\cite{CHAUVE-RECOMB09}.  However, heuristics
based on the computation of successive minimum edge-cuts in a graph
whose vertices are the considered species have been widely
used~\cite{SEMPLE-DAM105,PAGE-WABI02}. In such heuristics, each
minimum edge-cut, that splits the set of considered species in two
subsets corresponds to a speciation that results in two lineages. A
complete species tree results then from a sequence of such
speciations, each obtained from a minimum edge-cut.

In the present work we follow this path, and we attack the following
parsimony problem: given a set of gene trees, find a bipartition of
the considered genomes into two subsets, corresponding to a
speciation, that minimizes the number of duplications that happened
before this speciation. We call this problem the Minimum Duplication
Bipartition Problem. Although a restricted version of the more general
Minimum Duplication Problem, it leads, as for supertrees, to a natural
greedy heuristics to reconstruct a species tree from a set of gene
trees. Our main result is a polynomial time and space 3-approximation
algorithm for the the Minimum Duplication Bipartition Problem. Our
algorithm relies on a well-known generalization of the Minimum
Edge-Cut Problem, Submodular Function Minimization~\cite{FUJISHIGE91}.

We first define, in Section~\ref{sec:prelim}, gene trees, species
trees and duplications, then the Minimum Duplication Problem and
Minimum Duplication Bipartition Problem. In Section~\ref{sec:mec} we
show how the Minimum Duplication Bipartition Problem can be described
both in terms of prefix of the gene trees and of a variant of the
classical Minimum Edge-Cut Problem, namely the Minimum
Labeled-Edge-Cut Problem. This problem is NP-hard in general, but we
show in Section~\ref{sec:sfm} how a variant can be solved via
Submodular Function Minimization, that leads to an approximation
algorithm for the Minimum Duplication Bipartition Problem, with ratio
two times the optimal plus one: if an optimal first speciation implies
$k$ gene duplications, our algorithm returns, in polynomial time and
space, a bipartition that implies at most $2k+1$ duplications. In
Section~\ref{sec:all}, we describe some properties of the set of all
optimal bipartitions.

\section{Preliminaries}
\label{sec:prelim}

\paragraph{Gene and species trees.}
Let $\mathcal{G} = \{1,2,\ldots,k\}$ be a set of integers representing
$k$ different genomes (species).  A \emph{species tree} on
$\mathcal{G}$ is a tree with exactly $k$ leaves, where each $i \in
\mathcal{G}$ is the label of a single leaf. A tree is binary if every
internal vertex has exactly two children. A \emph{gene tree} on
$\mathcal{G}$ is a binary tree where each leaf is labeled by an
integer from $\mathcal{G}$. A gene tree is a formal representation of
a phylogenetic tree of a gene family, where each leaf labeled $i$
represents a member (gene) of the gene family located on genome $i$.

Given a vertex $x$ of a binary tree $T$, we denote by $L(x)$
(resp. $L(T)$) the subset of $\mathcal{G}$ defined by the labels of
the leaves of the subtree of $T$ rooted in $x$ (resp. the labels of
the leaves of $T$).  We denote by $x_l$ and $x_r$ the two children of
$x$ if $x$ is not a leaf.

\paragraph{The Minimum Duplication Problem.}
Given a gene tree $G$ and a, possibly non-binary, species tree $S$ on
$\mathcal{G}$, the {\em LCA mapping} $M$ maps vertices of $G$ to
vertices of $S$ as follows: for a vertex $x$ of $G$, $M(x)=v$ is the
unique vertex of $S$ such that $L(x) \subseteq L(v)$ and $v$ is a leaf
of $S$ or $L(x)$ is not included in the leaf set of any child of $v$. A
vertex $x$ of $G$ is then a {\em duplication with respect to $S$} if
$M(x) = M(x_r)$ and/or $M(x) = M(x_l)$; otherwise, $x$ is called a
{\em speciation} (see Appendix for more details on the evolutionary
events implied by the LCA mapping). A duplication $x$ of $G$ is said
to {\em precede the first speciation} if $M(x) = r(S)$, the root of
$S$. The same definitions apply to a forest $F$ of gene trees on
$\mathcal{G}$.  The \emph{duplication cost} of $F$ given $S$ denoted by
$d(F,S)$ is the number of vertices in $F$ that are duplications with
respect to $S$. It is well known that $d(F,S)$ is the minimum number
of gene duplication events required in any evolution scenario that
resulted in $F$ (see~\cite{GORECKI06,CHAUVE-RECOMB09} and references
there), which leads to the following optimization problem.

\smallskip
\noindent{\sc Minimum Duplication Problem (MDP)}:\\ 
{\bf Input:} A gene tree forest $F$  on  $\mathcal{G}$;\\ 
{\bf Output:} A binary species tree $S$ on $\mathcal{G}$ such that $d(F,S)$ is minimum.
\smallskip

The Minimum Duplication Problem with multiple gene trees is
NP-complete~\cite{MA-SIAMJC30}, and there is no known approximation
algorithm for this problem, although a fixed-parameter tractable
algorithm has been proposes in~\cite{HALLETT-RECOMB00}.

\paragraph{A restriction to the first speciation: the Minimum
  Duplication Bipartition Problem.}  A \emph{bipartition} $B$ on
$\mathcal{G}$ is a species tree on $\mathcal{G}$ containing only three
internal vertices, the root $v$ and its children $v_l$ and $v_r$, and
such that $L(v_l) \cap L(v_r) = \emptyset$ ($v_l$ and $v_r$ are
possibly non-binary vertices). From now, we always denote by $v$ the
root of a bipartition $B$. As the way vertices of a gene tree forest
$F$ were labeled as duplication or speciation for a given species tree
did not require that this species tree is binary, they apply without
changes to bipartitions, as do the notion of duplication in $F$
preceding the first speciation of $B$. We denote by $d_1(F,B)$ the
number of duplications that precedes the first speciation with respect
to a bipartition $B$. This leads to the following optimization
problem, that is a restriction of MDP
where the parsimony assumption is restricted to the duplications that precedes the
first speciation:

\smallskip
\noindent{\sc Minimum Duplication Bipartition Problem (MDBP)}:\\
{\bf Input:} A gene tree forest $F$ on $\mathcal{G}$;\\
{\bf Output:} A bipartition $B$ on $\mathcal{G}$ such that $d_1(F,B)$  is minimum.
\smallskip

The motivation for this problem follows a remark
in~\cite{CHAUVE-RECOMB09} that MDP is in fact a slight variant of
a supertree problem (see also~\cite{SCORNA-LATA09} that explores the
link between gene duplications and supertrees), and the fact that
greedy heuristics for hard supertree problems based on computing
successive speciations events have proved to be
effective~\cite{SEMPLE-DAM105,BININDA-EDMONDS-2004}. Indeed, given a
bipartition $B$ for a forest $F$ of gene trees, removing all vertices
$x$ such that $L(x)$ contains leaf labels both from $L(v_l)$ and
$L(v_r)$ results in two sets of trees: one set of trees, $F_l$, with
leaves that belong only to $L(v_l)$ and one, $F_r$, with leaves from
$L(v_r)$. These two forests of gene trees can then be considered
independently similarly to $F$, which defines a greedy heuristic to
compute a binary species tree
(see~\cite{SEMPLE-DAM105,PAGE-WABI02,CHAUVE-RECOMB09}).

We conclude these preliminaries with three obvious, but very useful,
properties related to duplication vertices of a forest of gene trees
$F$.

\begin{property}\label{prpy:prelim}
  Let $F$ be a gene trees forest on $\mathcal{G}$ and $x$ a vertex of
  $F$.
  \begin{enumerate}
  \item If $L(x_l) \cap L(x_r) \neq \emptyset$ is a duplication vertex
    with respect to any species tree $S$ (including bipartitions).
    Such a vertex is called an \emph{apparent duplication}.
  \item Given a bipartition $B$ on $\mathcal{G}$, with root $v$, $x$
    is a duplication with respect to $B$ that precedes the first
    speciation if and only if there exists a pair $(s,t)\in
    L(v_l)\times L(v_r)$ such that $(s,t) \in L(x_l)^2$ or $(s,t) \in
    L(x_r)^2$.
  \item Given a bipartition $B$ on $\mathcal{G}$, if $x$ is a
    duplication with respect to $B$ that precedes the first
    speciation, then every ancestor of $x$ is too.
  \end{enumerate}
\end{property}
 
\section{Prefix of gene trees and Minimum Labeled-Edge Cut}
\label{sec:mec}


\subsection{Minimum Labeled-Edge-Cut.}
Given a connected graph $H=(V,E)$, an \emph{edge-cut} of $H$ is an
edge set $E' \subseteq E$ whose removal disconnects the graph $H$, and
the \emph{size} of an edge-cut $E'$ is the number of edges in $E'$.  A
bipartition $B$ on the vertices of $H$ induces a unique edge-cut of
$H$ denoted by $E(B)$ and composed of the edges $(s,t) \in E$ such
that $s \in L(v_l)$ and $t \in L(v_r)$.  The \emph{Minimum Edge-Cut
Problem} is to find of a bipartition on the vertices of $H$ inducing
an edge-cut of $H$ of minimum size. It can be solved in linear
time~\cite{STOER-JACM44}. 

If the edges of $H$ are labeled on a given alphabet $A$, the
\emph{label-size} of an edge-cut $E'$ of $H$ is the size of the subset
of $A$ defined by the labels of the edges in $E'$. The following
problem is a natural variant of the Minimum Edge-Cut Problem:

\smallskip
\noindent{\sc Minimum Labeled-Edge-Cut Problem}:\\
{\bf Input:} A connected edge-labeled graph $H=(V,E)$;\\
{\bf Output:} A bipartition $B$ on $V$ such that  the label-size of 
the edge-cut $E(B)$ of $H$ induced by $B$ is minimum.
\smallskip

We now show how to reduce MDBP to a Minimum 
Labeled-Edge-Cut Problem. Let $F$ be a gene tree
forest with $m$ internal nodes. We label each internal vertex of $F$
using a unique integer of $A=\{1,\ldots,m\}$: no two internal vertices
can have the same label. We then associate to $F$ an edge labeled
graph $H(F)=(V,E)$ as follows:
\begin{itemize}
\item the set $V$ of vertices of $H(F)$ is the set $L(F)$ of labels of
leaves of $F$,
\item there is an edge between vertices $s$ and $t$, labeled by $a\in
  A$ if and only if the unique internal vertex $x$ of $F$ labeled by
  $a$ is such that $(s,t) \in L(x_l)^2$ or $(s,t) \in L(x_r)^2$.
\end{itemize}

\begin{lemma}
  \label{lem:reduction_melc}
  Let $F$ be a gene tree forest on $\mathcal{G}$ and $H(F)=(V,E)$ the
  edge-labeled graph associated to $F$. If $B$ is a bipartition on
  $L(F)$ then the cost $d_1(F,B)$ of $B$ is also the label-size of the
  edge-cut $E(B)$ of $H$ induced by $B$.
\end{lemma}

\begin{proof}
  If $E(B)$ contains an edge $(s,t)$ labeled with $a\in A$, then,
  from Property 1.2, the vertex of $F$ labeled by $a$ is a duplication
  that precedes the first speciation. Conversely, if  $a$ is the label of  
  duplication that precedes the first speciation then there is
  an edge $(s,t)$ in $E(B)$ labeled with $a$. 
\qed
\end{proof}


We conclude this section by some facts on the complexity of the
Minimum Labeled-Edge-Cut Problem. It is naturally linked to another
labeled edge-cut problem, the Minimum Label-Cut problem. Given a
connected edge-labeled graph $H=(V,E)$ with edges labeled on a set of
labels $A$, a \emph{label-cut} of $H$ is a subset $A'$ of $A$ such
that the removal of all edges of labels in $A'$ disconnects $H$ and
the size of $A'$ is the number of labels contained in $A'$. The
Minimum Label-Cut problem is then defined as follows:

\smallskip
\noindent{\sc Minimum Label-Cut  problem}:\\
{\bf Input:} A connected edge-labeled graph $H=(V,E)$;\\
{\bf Output:} A label-cut $A'$ of $H$ whose size is minimum.

\begin{lemma}
  \label{lem:eq}
  An edge-labeled graph has a Minimum Labeled-Edge-Cut of label-size $k$ if
  and only if it has a Minimum-Label-Cut of size $k$.
\end{lemma}

Lemma \ref{lem:eq} follows obviously from the definition of the two
problems. It has been shown that the Minimum Label-Cut Problem is
NP-hard when a pair of vertices that should be separated by the
edge-cut induced by $A'$ is
given~\cite{JHA-LABELCUT-02}. Lemma~\ref{lem:eq} implies that the same
holds for the Minimum Labeled-Edge-Cut Problem. Note however that even
the hardness of the general Minimum Labeled-Edge-Cut Problem would not
directly imply the hardness of MDBP as not any graph can be the graph $H(F)$ obtained
from a gene trees forest $F$. 


\subsection{Prefixes of gene tree forests.}
We now describe an alternative point of view on parsimonious first
speciations, that does not consider the graph $H(F)$, but directly the
gene trees. 

A \emph{prefix} of a tree is a set $I$ of vertices such that, if $a
\in I$, then every ancestor of $a$ belongs to $I$. From Property 1.3,
given a gene tree forest $F$, a bipartition $B$ on $L(F)$, with root
$v$, induces a unique prefix of $F$ composed of the set of
duplications in $F$ that precedes the first speciation. Conversely, a
prefix $I$ of $F$ induces a partition $P(I)$ of $L(F)$ as follows: two
species are in the same part of $P(I)$ if and only if they belong to
the same connected component in the graph obtained from $H(F)$ by
removing all edges corresponding to vertices in the prefix. $|P(I)|$
denotes the number of parts of $P(I)$. We now introduce an
optimization problem related to prefixes of gene tree forests.

\smallskip
\noindent{\sc Minimum Duplication Prefix Problem (MDPP)}:\\
{\bf Input:} A gene tree forest $F$ on $\mathcal{G}$;\\
{\bf Output:} A minimum size prefix $I$ of $F$ such that $|P(I)|\geq 2$.
\smallskip

\begin{lemma}
  Given a gene tree forest $F$ and the edge-labeled graph $H(F)$
  associated to $F$, the minimum label-size of an edge-cut of $H(F)$
  is equal to the minimum size of a prefix $I$ of $F$ such that
  $|P(I)|\geq 2$.
\end{lemma}
\begin{proof}
  An edge-cut $E'$ of $H(F)$ induces a bipartition on $L(F)$ inducing
  a prefix $I$ of $F$ such that $|P(I)|\geq 2$ and the size of $I$ is
  the label-size of $E'$. Conversely, given a prefix $I$ of $F$ such
  that $|P(I)|\geq 2$, any bipartition $B=(L_1,L_2)$ of $L(F)$ such
  that any part of $P(I)$ is included either in $L_1$ or in $L_2$
  induces an edge-cut $E(B)$ of $H(F)$ whose label-size is less or
  equal to the size of $I$.  
  \qed
\end{proof}

\section{A polynomial 3-approximation via submodular function minimization}\label{sec:sfm}

We show now how, by defining, from $F$, a slightly different graph
than $H(F)$, the Minimum Labeled-Edge-Cut Problem can be solved, for
such graphs, in polynomial time and give a 3-approximation for MDBP.

\subsection{Cut-set and submodular function}
\label{ssec:cut-set}

A \emph{submodular function} is a set function $f: 2^V \rightarrow
\mathbb{R}$ defined from the subsets of a finite set $V$ to the set of
real numbers $\mathbb{R}$ such that for any subsets $A$ and $B$ of
$V$, $f(A) + f(B) \geq f(A \cup B) + f(A \cap B)$. The set $V$ is then
called the \emph{ground set} of $f$.\\ Several combinatorial
optimization problems have been linked to submodular
functions~\cite{FUJISHIGE91}. Given a submodular function $f$, the
following optimization problem, that can be solved using                                                                          
combinatorial polynomial time algorithms~\cite{IWASA-SODA09}, is often
considered:

\smallskip
\noindent{\sc Submodular Function Minimization (SFM) Problem}:\\
{\bf Input:} A submodular function $f: 2^V \rightarrow \mathbb{R}$ with ground set $V$;\\
{\bf Output:} A subset $V'$ of $V$ such that $f(V')$ is minimum.
\smallskip

The Minimum Edge-Cut problem is a special case of SFM: given a graph
$H=(V,E)$, if $g(X)$ denotes the size of the edge-cut $E(B)$ induced
by the bipartition $B$ on $V$ such that $L(v_l)=X$ and $L(v_r)=V-X$,
then the cut-set function $g$ is a submodular function. The problem of
minimizing $g$ is then the Minimum Edge-Cut problem on $H$.

The Minimum Labeled-Edge-Cut problem can also be reduced to a Set
Function Minimization: given an edge-labeled graph $H=(V,E)$, we
associate to $H$ the cut-set function $f(H): 2^V \rightarrow
\mathbb{R}$ defined from the subsets of $V$ to $\mathbb{R}$ such that,
for any subset $X$ of $V$, $(f(H))(X)$ is the label-size of the
edge-cut $E(B)$ induced by the bipartition $B$ on $V$ such that
$L(v_l)=X$ and $L(v_r)=V-X$. It is then easy to see that solving the
Minimum Labeled-Edge-Cut problem on $H$ can be achieved by minimizing
$f(H)$. However, the cut-set function induced
by the edge-labeled graph associated to a gene tree forest $F$ is not
always submodular, as we show now. 

We describe now a property of the cut-set function for the Minimum
Labeled-Edge-Cut Problem that will be crucial in designing an
approximation algorithm for MDBP. Given two subsets $A$ and 
$B$ of $L(F)$, we define the following sets of labels:
\begin{itemize}
\item $AB$ is the set of labels of edges $(s,t)$ in $H(F)$ such that $s\in A-B$ and $t\in B-A$,
\item $C_1$ is the set of labels of edges $(s,t)$ in $H(F)$ such that $s\in A\cap B$ anthe d $t\not\in A\cup B$, 
\item $A_1$ is the set of labels of edges $(s,t)$ in $H(F)$ such that $s\in A-B$ and $t\not\in A\cup B$, 
\item $B_1$ is the set of labels of edges $(s,t)$ in $H(F)$ such that $s\in B-A$ and $t\not\in A\cup B$,
\item $AC$ is the set of labels of edges $(s,t)$ in $H(F)$ such that $s\in A\cap B$ and $t\not\in B-A$, 
\item $BC$ is the set of labels of edges $(s,t)$ in $H(F)$ such that $s\in A\cap B$ and $t\not\in A-B$. 
\end{itemize}

\begin{lemma}\label{lem:notsubmodular}
If the cut-set function $f'=f(H(F))$ induced by $H(F)$ is not a
submodular function then there exists at least one internal vertex $x$
of $F$ labeled $a$ and two subsets $A$ and $B$ of $L(F)$ such that:
\begin{itemize}
\item (1) $a \in A_1 \cap AC$ and $ a \not\in AB \cup C_1 \cup BC \cup B_1$ or  
\item (2) $a \in B_1 \cap BC$ and $ a \not\in AB \cup C_1 \cup AC \cup A_1$.
\end{itemize}
and $x$ is necessarily a node which is not an apparent duplication.
\end{lemma}

\begin{proof}
$f'(A) + f'(B) = |AB \cup C_1 \cup A_1 \cup AC| + |AB \cup C_1 \cup
B_1 \cup BC|$.\\ $f'(A\cup B) + f'(A\cap B) = |C_1 \cup A_1 \cup B_1|
+ |C_1 \cup AC \cup BC|$.\\ Any label $a$ of a vertex of $F$ that
belongs to more than two or only one of the six sets ($A_1$, $B_1$,
$C_1$, $AB$, $AC$, $BC$) contributes in $f'(A) + f'(B)$ with a weight
that is greater or equal to its contribution in $f'(A\cup B) +
f'(A\cap B)$.  If $a$ belong to exactly two of the sets, then its
contribution $C_1$ in $f'(A) + f'(B)$ is at least its contribution
$C_2$ in $f'(A\cup B) + f'(A\cap B)$ except when these sets are $A_1$
and $AC$, or $B_1$ and $BC$. In these case $C_1=1$ and $C_2=2$.\\ The
label of vertex $x$ can not belong to only two sets if $x$ is an
apparent duplication.
\qed
\end{proof}

\subsection{ A submodular modification}
\label{ssec:modification}

Lemma \ref{lem:notsubmodular} suggests a modification of the
definition of the edge-labeled graph $H(F)$ associated to a gene tree
forest $F$ such that properties (1) and (2) never hold, which leads to
a cut-set function that is submodular.  Given a gene tree forest $F$
on $ \mathcal{G}$ with a labeling of its internal vertices using a set
of label $A$, we now associate to $F$, the graph $I(F)=(V,E)$
defined as follows:
\begin{itemize}
\item $V=L(F)$,
\item there is an edge $(s,t)$ in $E$, labeled with $a\in A$ if and
  only if there exists a vertex $x$ in $F$ such that
  \begin{itemize}
  \item either $(s,t) \in L(x_l)^2$ or $(s,t) \in L(x_r)^2$ 
  \item or $(s,t) \in L(x)^2$ and $x$ is not an apparent duplication.
  \end{itemize}
\end{itemize}

\begin{lemma}\label{lem:submodular}
Given a gene tree forest $F$, the cut-set function $f(I(F))$ associated to $I(F)=(V,E)$  is a 
submodular function.
\end{lemma}

\begin{proof}
Properties (1) and (2) of Lemma \ref{lem:notsubmodular} never hold.\qed
\end{proof}

It follows from Lemma \ref{lem:submodular} that the minimization of $f(I(F))$ can be solved
in polynomial time. The following theorem shows that it provides a 3-approximation algorithm 
for MDBP. 

\begin{theorem}
Given a gene tree forest $F$ with $n$ vertices, on a set of $k$
genomes, and the modified graph $I(F)=(V,E)$ associated to $F$, if $c$
is the minimum value of the cut-set function $f'=f(I(F))$ and $d$ is
the minimum cost $d_1(F,B)$ associated to a bipartition $B$ on 
$V=L(F)$, then $c \leq 2*d+1$. Moreover, $c$ can be computed 
in $O(k^6n\log(kn))$ time.
\end{theorem}

\begin{proof}
Let $d$ be the minimum cost of a bipartition on $V=L(F)$ and let us
consider a bipartition $B$ on $V$ such that $d_1(F,B)=d$.  If
$X=L(v_l)$, then by definition $f'(X)$ is the number of vertices $x$
in $F$ such that there exists a couple $(s,t)\in X\times (V-X)$ such
that:

\begin{description}
\item (1) $(s,t)\in L(x_l)$ or $(s,t)\in L(x_r)$ or
\item (2) $(s,t)\in L(x)$ and $x$ is not an apparent duplication.
\end{description}
 
Moreover, any vertex $y$ of $F$ that is a strict ancestor of a vertex
satisfying (1) or (2), satisfies (1).  Then, the vertices of $F$
satisfying (1) or (2) form a prefix $I$ of $F$ such that only leaves
of $I$ can satisfy (2).  This induces that if $p$ is the number
vertices of $F$ satisfying (1) then $f'(X) \leq 2*p+1$ (since the
number of leaves of a binary tree having $i$ internal vertices is
$i+1$).  Finally, by definition, the cost $d_1(F,B)$ of $B$ is equal to
$p$ and then $p=d$.  

The complexity follows from the algorithm described in~\cite{IWASA-SODA09}.
\qed
\end{proof}

\section{The set of all optimal solutions}
\label{sec:all} 


A common approach in the Minimum Edge-Cut based algorithms used for
supertrees problem is to seek not only a single parsimonious
bipartition on the set of genomes, but all possible ones. The problem
is then finding not one optimal bipartition on the set of species, but
a partition compatible with all optimal
bipartitions~\cite{SEMPLE-DAM105,PAGE-WABI02}.  

Given a gene tree forest $F$, the partition of $L(F)$ compatible with
all optimal solutions of MDBP is such that two elements of $L(F)$ are
in the same part if and only if they belong to the same part in all
optimal bipartitions.  This partition is denoted by $\mathcal{PB}(F)$.
From the point of view of gene trees prefixes, there can also be
several optimal solutions for MDPP.  The unique partition of $L(F)$
compatible with all optimal solutions of MDPP is such that two
elements of $L(F)$ are in the same part if and only if they belong to
the same part in all partitions induced by optimal prefixes.  This
partition is denoted by $\mathcal{PP}(F)$.

\begin{proposition}\label{prop:equivalence}
  Given a gene tree forest $F$, $\mathcal{PB}(F)=\mathcal{PP}(F)$.
\end{proposition}
\begin{proof}
  This result is straightforward consequence of the definition of
  $\mathcal{PB}(F)$ and $\mathcal{PP}(F)$ since the set of edges that
  belong to a minimum label-size edge-cut of $H(F)$ is exactly the set
  of edges that have a label belonging to a minimum size label-cut of
  $H(F)$.  \qed
\end{proof}

For the classical Minimum Edge-Cut problem, the problem of computing
the unique partition of the set of vertices of a graph compatible
with all minimum edge-cuts can be reduced to a decision problem on
edges of the graph that can be solved
efficiently~\cite{STOER-JACM44}. We generalize it below, and give also
an equivalent formulation in terms of gene trees prefixes:

\smallskip
\noindent{\sc Minimum Duplication Bipartition Edge Problem}:\\
{\bf Input:} A graph $G$ and an edge $e$ of $G$;\\
{\bf Output:} Does $e$ belong to a minimum label-size edge-cut of $G$.

\smallskip
\noindent{\sc Minimum Duplication Prefix Vertex Problem}:\\
{\bf Input:} A gene tree forest $F$ and a vertex $a$ of $F$;\\
{\bf Output:} Does $a$ belong to a minimum size prefix $I$ of $F$ such that  $|P(I)|\geq 2$.

\begin{lemma}\label{reduction-decisional}
  Given a gene tree forest $F$, computing $\mathcal{PB}(F)$ can be
  reduced to solving the Minimum Duplication Bipartition Edge Problem
  for each edge $(s,t)$ of $H(F)$ or the Minimum Duplication Prefix
  Vertex Problem for each vertex $a$ of $F$.
\end{lemma}
    
\begin{proof}
  The partition $\mathcal{PB}(F)$ is such that two elements of $L(F)$
  are in the same part if and only if they belong to the same
  connected component in the graph obtained from $H(F)$ after removing
  all edges belonging to a minimum label-size edge-cut of $H(F)$. The
  equivalence with the Minimum Duplication Prefix Vertex Problem
  follows immediately from Proposition~\ref{prop:equivalence}.  \qed
\end{proof}



\section{Conclusion}
\label{sec:conc}

We showed that computing a parsimonious first speciation in the gene
duplication model can be approximated in polynomial time with a ratio
of $3$. As far as we know this is the first time a constant
approximation algorithm is proposed in relation with the problem of
inferring species trees using gene duplications. This result was
obtained by describing it in terms of edge-cuts in particular graphs,
that can be computed in polynomial time through submodular functions
minimization.

The complexity status of the Minimum Duplication Bipartition Problem
is still open, but its relationship to the Minimum Label-Cut Problem,
together with the fact the corresponding cut-set function is not
submodular seems to indicate it is NP-hard. As for classical minimum
edge-cut, this question is strongly related to the question of the
complexity of deciding if a single edge of a graph belongs to a
minimum labeled-edge-cut, or if an edge label belongs to a minimum
label-cut. From an algorithmic point of view, there is still a large
gap between the near linear time complexity of the simple Minimum
Edge-Cut Problem and the high complexity of our version of the Minimum
Labeled-Edge-Cut Problem. In order to make the approximation algorithm
we propose useful in a practical context, some advances on the
complexity of the Minimum Labeled-Edge-Cut Problem are necessary.

Finally, it is important to remark that all these questions have been
described in terms of edge-cuts in a graph defined by a gene tree
forest, but not all graphs can be induced by a gene trees forest.
Hence, for our particular purpose, if one hopes to show the problems
we introduced are tractable, it is probable that we should use
specific properties of these graphs.  Alternatively, using prefixes of
gene trees could be a way to attack these questions.

\smallskip
\noindent{\bf Acknowledgements.} Work supported by an NSERC Discovery
grant to C.C. and a fellowship from the ANR BRASERO project
(ANR-06-BLAN-0045) for A.O. We thank Tamon Stephen for pointing to the
link with submodular function minimization.


\newpage 

\section*{Appendix}

\paragraph{Reconciliation.} A \emph{subtree insertion} in a forest $F$ consists in
grafting a new subtree onto an existing branch of $F$. An 
\emph{extension} of $F$ is a tree which can be obtained from $F$ by 
subtree insertions on the branches of $T$.

A gene tree forest $F$ is said to be \emph{DS-consistent} with a
binary species tree $S$ on $\mathcal{G}$ if, for every vertex $x$ in
$F$ such that $|L(x)\geq 2|$, there exists a vertex $v$ in $S$ such
that $L(x)=L(v)$ and one of the following conditions holds:
\begin{description} 
\item (D) either $L(x_l)=L(x_r)$ (indicating a Duplication),
\item (S) or $L(x_l)=L(v_l)$ and $L(x_r)=L(v_r)$ or inversely
  (indicating a Speciation).  
\end{description} 

A \emph{reconciliation} between a gene tree forest $F$ and 
a binary species tree $S$ on on $\mathcal{G}$ is an extension $R(F,S)$
of $F$ that is DS-consistent with $S$.  A reconciliation 
between $F$ and $S$ implies an unambiguous evolution scenario for the 
gene family $F$ where a vertex of $R(F,S)$ that satisfies property (D) 
represents a duplication, and an inserted subtree represents a gene 
loss. Vertices of $R(F,S)$ that satisfy property (S) represent speciation events.  
It is immediate to see that every vertex $x$ of $T$ such that 
$L(x_l) \cap L(x_r) \neq \emptyset$  will always be a duplication vertex in 
any reconciliation $R(F , S )$ between $F$ and $S$. 

\paragraph{Optimization problems.} The following cost measure is
considered for a reconciliation $R(F,S)$ between a gene tree forest $F$ and 
a species tree $S$: the \emph{duplication cost} of $R(F,S)$ denoted by
  $d(R(F,S),S)$ is the number of duplications induced by $R(F,S)$.
 When a species tree is not known, the following  natural combinatorial
  optimization problems is often considered.\\

\noindent{\sc Minimum Duplication Problem I}:\\ 
{\bf Input:} A gene tree forest $F$  on  $\mathcal{G}$;\\ 
{\bf Output:} A binary species tree $S$ such that $d(R(F,S),S)$ is minimum.\\
 
Given a gene tree forest $F$ and a species tree $S$  on 
$\mathcal{G}$, the LCA mapping between $F$ and $S$ induces a 
reconciliation between $F$ and $S$ where an internal 
vertex $x$ of $F$ leads to a duplication vertex if 
$M (x_l ) = M (x)$ and/or $M (x_r ) = M (x)$. 
In \cite{GORECKI06}, it has been shown that the reconciliation $M(F,S)$ between 
$F$ and $S$ defined by their LCA mapping  minimizes the duplication cost.

\end{document}